\documentclass[%
 reprint,
superscriptaddress,
 amsmath,amssymb,
 aps,
 prl
]{revtex4-2}

\makeatletter
\let\frontmatter@footnote@produce\frontmatter@footnote@produce@footnote
\makeatother

\usepackage{braket}
\usepackage{tikz}
\usepackage{adjustbox}

\usepackage{graphicx}
\usepackage{dcolumn}
\usepackage{bm}

\usepackage[colorlinks=true, citecolor=blue, linkcolor=blue, urlcolor=blue]{hyperref}
\hypersetup{breaklinks=true}
\definecolor{myred}{HTML}{a82a2a}
\definecolor{myblue}{RGB}{30,80,190}    
\definecolor{corrred}{HTML}{A82A2A}     
\definecolor{causbrown}{HTML}{5C3317}   
\newtheorem{theorem}{Theorem}
\usepackage{mdframed}
\surroundwithmdframed[innertopmargin=6pt, innerbottommargin=6pt,
    innerleftmargin=8pt, innerrightmargin=8pt, skipabove=8pt,
    skipbelow=8pt, nobreak=true]{theorem}
\newcommand{\qed}{\hfill\ensuremath{\square}}
\usepackage{dsfont}

\begin{document}

\title{Correlation versus Causation in Quantum Criticality}

\author{Conrad Wichmann}
\email{cwichmann@fas.harvard.edu}
\affiliation{Department of Physics, Harvard University, Cambridge, Massachusetts 02138, USA}
\affiliation{Pritzker School of Molecular Engineering, University of Chicago, Chicago, Illinois 60637, USA}

\author{Ryan Thorngren}
\affiliation{Mani L. Bhaumik Institute for Theoretical Physics, Department of Physics and Astronomy, University of California, Los Angeles, California 90095, USA}

\author{Ruben Verresen}
\affiliation{Pritzker School of Molecular Engineering, University of Chicago, Chicago, Illinois 60637, USA}

\date{\today}

\begin{abstract}

Correlation functions $\braket{O_1(x)\, O_2(0)}$ reveal scaling dimensions through spatial decay. We instead consider static susceptibility, the change in $\braket{O_1(x)}$ from perturbing the Hamiltonian by $O_2(0)$, which we term causation for short. In a conformal field theory (CFT), dimensional analysis predicts decay of $|x|^{-2\Delta}$ for correlation and $|x|^{-2\Delta+1}$ for causation. Yet we find causation can decay up to fifteen additional orders in $x$ through a general mechanism, which we trace to time-derivative fields being unable to contribute to static response. In higher-dimensional CFTs, this mechanism ensures leading causation arises from primaries, even when descendants dominate correlation, which we leverage with DMRG to identify a previously unresolved corner primary of $\Delta \approx 8.8$ and a heavy magnetic line defect primary of $\Delta \approx 4.6$ in the $(2+1)$D critical Ising model. Moreover, the same mechanism governs edge-mode localization in $(1+1)$D gapless symmetry-protected topological phases, explaining previously observed anomalously small edge-mode splittings and guiding our construction of spin chains with splittings as small as $1/L^{18}$ and $1/L^{25}$.

\end{abstract}

\maketitle

\textit{Introduction---}Correlation, a \textit{passive} measure of statistical dependence between observables in physical systems, is ubiquitous in quantum field theories. Within quantum many-body physics in particular, correlation functions are of significant interest because their spatial decay characterizes the system's phase and other universal properties \cite{Sachdev2011, Wen2004, Fradkin2013}. Yet many-body physics is often concerned with \textit{active} behavior: how one observable responds when the Hamiltonian is perturbed by another, as encoded in the static susceptibility \cite{Kubo1957, Kubo1966, Polkovnikov2011, Forster1975, Giuliani2005}. In this work, we investigate how static susceptibility scales at criticality and show that, much like the correlation function, it encodes information characterizing the corresponding universality class. Motivated by this parallel, we refer to this quantity as the `causation function.'

Given their closely related physical interpretations, one might expect strongly correlated observables to exert strong causal influence. The first key result of our work is that this expectation fails at criticality. There is a wide class of critical systems in which causation is far more suppressed than correlation: while correlations decay algebraically, causation can decay fifteen orders of magnitude faster, in some cases even exponentially \footnote{Here, we compare two equilibrium quantities, which differs from the known effect that the $q \to 0$ and $\omega \to 0$ limits of the dynamical susceptibility need not commute \cite{Forster1975}.}.

This suppression has significant implications for the stability of topological edge modes. A conformal field theory (CFT) can subdivide into distinct symmetry-enriched versions, many of which are gapless symmetry-protected topological phases (gSPTs) hosting sharply localized edge modes \cite{Scaffidi2017,Verresen2021, Kestner2011, Grover2012, Keselman2015, Parker2018, Verresen2018, Verresen2020, Borla2021, Thorngren2021, Duque2021, Hidaka2022, MaZouWang2022, Yu22, WenPotter2023, WenPotter2023b, LiOshikawaZheng2024, SuZeng2024, Yu2024, Prembabu2024, Prembabu25, Prembabu25b}. In the Majorana CFT, prior work computed exceptionally small finite-size splittings for these modes: exponentially so in the free case \cite{Verresen2018}, and as $1/L^{14}$ with interactions \cite{Verresen2021}. Our second key result provides a general explanation for these model-specific findings: edge mode localization is quantified by causation and hence inherits the bulk's strong causal decay. Guided by this mechanism, we construct two new gapless SPTs in quantum spin chains whose splittings, $1/L^{18}$ and $1/L^{25}$, are suppressed well beyond prior work.

Causation is especially informative in higher dimensions, where there has long been significant interest in the spectrum of $(2+1)$D CFTs \cite{Polyakov1970, Wilson1972, Polyakov1974, Rychkov2017, Henriksson2023} and their boundaries and defects \cite{GaiottoMazacPaulos2014,
Metlitski2022, Padayasi2022, ParisenToldin2022, ZhouGaiotto2024,
ZhouZou2025, Lanzetta2025}. Approaches include the conformal bootstrap \cite{Rattazzi2008, ElShowk2012, Kos2016, SimmonsDuffin2017, Poland2019} and, more recently, fuzzy sphere regularization \cite{Zhu2023, HeZhu2026}. While scaling dimensions can be extracted from lattice correlations \cite{Pelissetto2002, Hasenbusch2010}, microscopic operators can be dominated by descendant fields, so correlation need not expose primary dimensions \cite{Mong2014, Zou2020, Sandvik2024}. Remarkably, causation can filter these descendants out. Our third key result is that beyond one spatial dimension, the leading causation of a generic lattice operator \footnote{By generic, we refer to lattice operators without definite spatial parity.} arises only from primaries, suggesting causation functions are good probes of the operator spectrum. We demonstrate this in the $(2+1)$D critical Ising model, where the response to a corner perturbation, measured at the opposite corner and at a magnetic line defect, uncovers a corner primary of dimension $\Delta \approx 8.8$ and a defect primary of $\Delta \approx 4.6$, both invisible to correlation.

\textit{Causation---}The relationship between two observables, $A(x)$ and $B(y)$, in the quantum ground state is typically expressed through their `correlation,' $\bra{0} A(x) B(y) \ket{0}$. Here, we consider a different probe of this relationship, namely \textit{causation}: if we perturb a Hermitian system, $H$, as $H_\lambda = H + \lambda B(y)$, how does $A(x)$ respond in the perturbed ground state? Denoting the perturbed ground state by $\ket{0}_\lambda$, perturbation theory gives \cite{Kubo1957} \begin{align} \label{eq:causation_def}
    \lim_{\lambda \to 0} \partial_\lambda \braket{A(x)}_\lambda \propto \braket{A(x)  \; G  \; B(y)} + \text{h.c.}, 
\end{align}where $G = (E_0 - H)^{-1}$ is the resolvent \footnote{We assume $A$ and $B$ have vanishing ground-state expectation values and that the ground state is non-degenerate, such that Eq.~(\ref{eq:causation_def}) is well-defined.}. Hence, causation is quantified by the static Green's function, $\braket{A(x)  \; G  \; B(y)}$, which we term the causation function.

By dimensional analysis, correlation and causation are expected to exhibit similar scaling at criticality. Indeed, if $\Delta_A$, $\Delta_B$, and `$-1$' denote the CFT scaling dimensions of $A$, $B$, and $G$, respectively (more generally, $G$ has dimension $-z$ with $z$ the dynamical critical exponent), then correlation scales as $  1/|x-y|^{\Delta_A + \Delta_B}$ and causation as $  1/|x-y|^{\Delta_A + \Delta_B - 1}$ \cite{DiFrancesco1997}. Surprisingly, however, in time-reversal-symmetric systems, causation can be far more suppressed.

This suppression occurs because operators which are time derivatives of other operators exhibit vanishing causation. To see this, let $A(x) \propto \partial_t C(x) = i[H, C(x)]$ be a time derivative of a third operator, $C(x)$, and assume all operators have vanishing expectation values. Then \begin{equation}
\begin{aligned} \label{eq:time_descendant} 
\braket{A\, G\, B} + \text{h.c.}
&= i \braket{ [H-E_0,C]\, G\, B} + \text{h.c.} \\
&= i \braket{ C(E_0-H)\, G\, B} + \text{h.c.} \\
&= i \braket{ [C(x),B(y)] } = 0,
\end{aligned}
\end{equation}where to get the third equality we used the result $(E_0 - H)G = \mathbb I - P_{gs}$, with $P_{gs}$ the ground state projector. The final expression vanishes whenever $|x-y|$ is sufficiently large that $C(x)$ and $B(y)$ commute due to disjoint support. Moreover, while the Hermitian combination vanishes generally, the individual causation function $\braket{A \; G \; B}$ vanishes when $H$ is time-reversal symmetric and both operators have the same $T$-charge \footnote{This renders $\braket{A \; G \; B}$ real; vanishing follows by anti-Hermiticity of $i \braket{C  B}$. }.

Now consider a microscopic model with lattice operator $O_i$ which flows in the continuum limit as $O_i \to \phi_1(x_i) + \phi_2(x_i) + \ldots$. Correlation and causation of $O_i$ are often dominated by different $\phi_i$. Correlation is controlled by the lowest-dimension field in the expansion, $\phi_\text{min}$, and decays as $|x-y|^{- 2 \Delta_{\phi_\text{min}}}$ \cite{Cardy1996}. By contrast, causation is controlled by the lowest-dimension field that is \textit{not} a time derivative, $\phi_\text{min}'$, and thus decays as $|x-y|^{- 2 \Delta_{\phi_\text{min}'} + 1}$. 

The suppression of causation is most dramatic for edge operators, where there are fewer non-time-derivative descendants. On a $(0+1)$D defect or boundary of a generic CFT, every descendant is either a time derivative or a quasiprimary \cite{Cardy2004}, so causation is controlled by the lowest-dimension \textit{quasiprimary}. In higher dimensions, quasiprimaries coincide with primaries \cite{DiFrancesco1997}, so causation of $(0+1)$D hinge or defect operators is controlled by the lowest-dimension \textit{primary} field.
\begin{figure}[t]
  \centering
  \includegraphics[width=0.9\columnwidth]{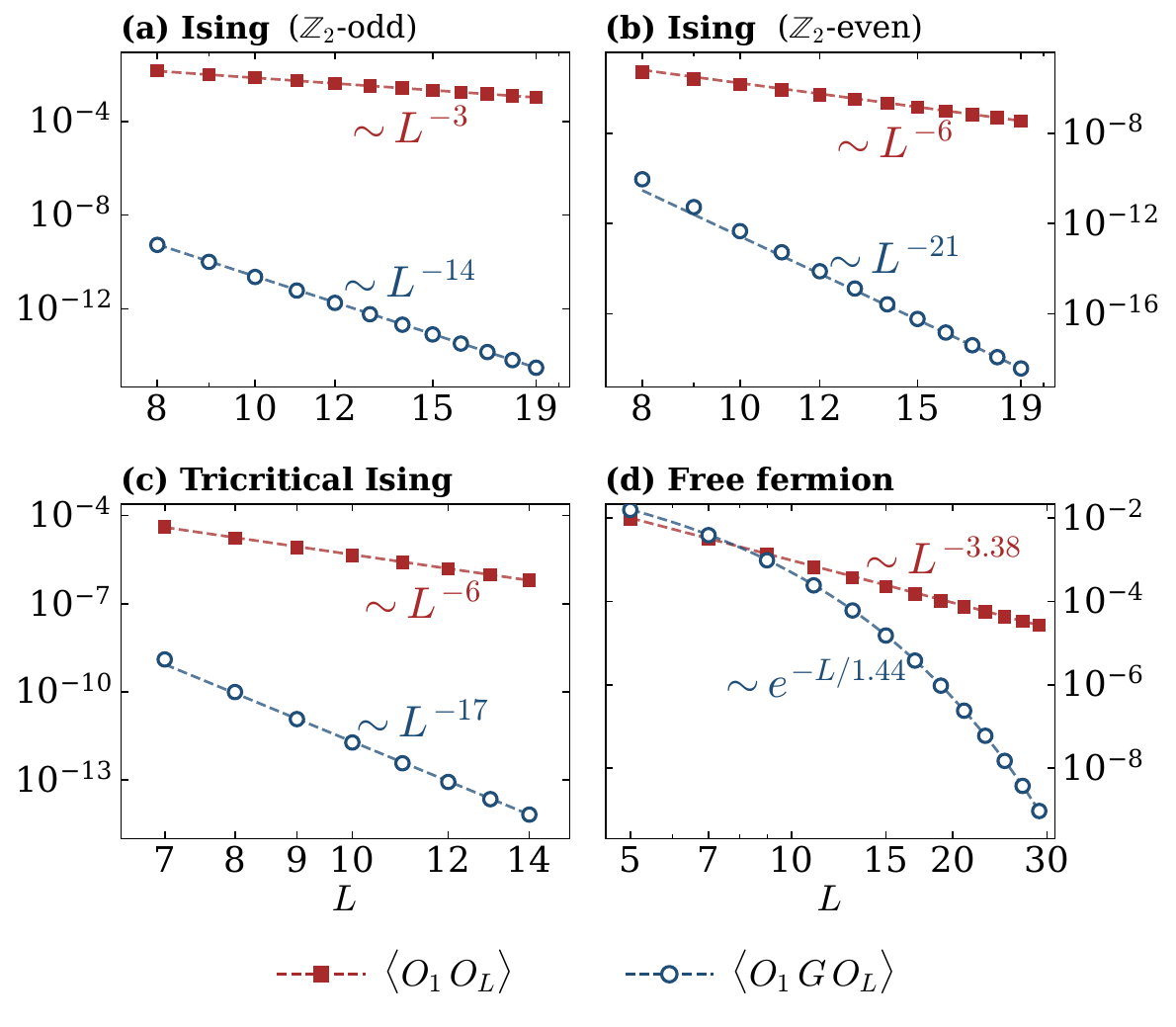}
  \vspace{-1.2em}
  \caption{
  \textbf{Correlation vs.\ causation in $(1+1)$D critical systems.} In each panel, the edge-to-edge correlation $\braket{O_1 O_L}$ and causation $\braket{O_1 \, G \, O_L}$ are computed between a $\mathbb Z_2^T$-odd operator $O_1$ at the left end and its mirror image $O_L$ at the right end. (a,b) Critical Ising with $O_1 = Y_1 X_2$ ($\mathbb{Z}_2$-odd) and $O_1 = Y_1 Z_2 X_3$ ($\mathbb{Z}_2$-even), respectively. (c) Tricritical Ising via the spin-1 Blume-Capel chain with $O_1 = S^y_1 S^z_2 S^x_3$. (d) BDI Majorana chain with parameters $(t_{-1}, t_0, t_1) = (1, 3, 2)$ and $t_\alpha = 0$ otherwise, with $O_1 = \tilde\gamma_1$ and $O_L = \gamma_L$. Panels (a,b,d) are obtained by free-fermion diagonalization; (c) by exact diagonalization in the spin basis.
  }
  \label{fig:figure1}
\end{figure}

\textit{Example: Ising CFT---}We first examine correlation and causation in the critical quantum Ising chain with free boundaries, \begin{align}
    H = - \sum_{n=1}^{L-1} Z_n Z_{n+1} - \sum_{n=1}^L X_n.
\end{align}This Hamiltonian has two symmetries: the $\mathbb Z_2$ Ising symmetry $P = \prod_n X_n$, and the antiunitary $\mathbb Z_2^T$ symmetry $T = K$, given by complex conjugation in the $Z$-basis. The $\mathbb Z_2$-odd, $T$-even operator $Z_n$ flows to the spin field $\sigma(x) + \cdots$, where the ellipsis denotes subleading fields, while the $\mathbb Z_2$-odd, $T$-odd operator $Y_n \propto i [H, Z_n] = \partial_t Z_n$ flows to the term-by-term time derivative of this expansion, $\partial_t \sigma(x) + \cdots$ \cite{Zou2020}. Correlation therefore decays as $\braket{Y_n Y_m} \sim 1/|n-m|^{2 \cdot (\Delta_\sigma+1)}$, with $\Delta_\sigma = 1/8$ in the bulk and $\Delta_\sigma = 1/2$ at the boundary \cite{Cardy1986}. Causation, by contrast, \textit{vanishes}: $ Y $ is an exact time derivative on the lattice, so $\braket{Y_n G Y_m} = 0$ by Eq.~(\ref{eq:time_descendant}).

A more generic operator with the same charges is no longer an exact time derivative and acquires subleading contributions. To predict these, we must identify the first $T$-odd field in its expansion that is a quasiprimary. For simplicity, we work at the boundary, where $\mathbb Z_2$-odd operators flow to the $h=\frac{1}{2}$ tower. Since $T$ charge alternates between descendant levels \footnote{This follows from $T L_{-n} T^{-1} = (-1)^n L_{-n}$ on the boundary. See End Matter for more details.}, a $T$-odd operator flows entirely to odd-level descendants of this tower. The corresponding Virasoro character is \cite{DiFrancesco1997} \begin{align}  \label{eq:ising_odd}
    \chi_{1/2} &\propto 1 + q + q^2 + q^3 + 2q^4 + 2q^5 + 3q^6 \nonumber\\
    &\quad  + 4q^7 + 5q^8 + 6q^9 + \cdots 
\end{align}where the coefficient of $q^n$ counts the number of descendant fields at level $n$. Since $\partial_t$ acts injectively from level $n-1$ to $n$ \footnote{Injectivity follows from unitarity. For a state of weight $h$, $\lVert L_{-1}\ket{\psi}\rVert^2 = 2h \lVert\ket{\psi}\rVert^2 + \lVert L_1\ket{\psi}\rVert^2$, where unitarity guarantees $h > 0$ and norms are non-negative for all states above the vacuum, so the right-hand side is strictly positive; the identity ($h = 0$) is the only exception, which is annihilated by $L_{-1}$.}, quasiprimaries exist at level $n$ if the preceding level has fewer fields, i.e., when the coefficient of $q^n$ is larger than that of $q^{n-1}$. Reading off the coefficients, quasiprimaries occur at levels $0, 4, 6, 7, \ldots$, with the first odd level at $n = 7$. Correlation is thus dominated by the lowest $T$-odd descendant at level 1 ($h = \frac{3}{2}$), while causation is dominated by the level-7 quasiprimary ($h = \frac{15}{2}$): for a boundary operator $O_1$ with these charges and its mirror image $O_L$ at the opposite end, correlation decays as $\braket{O_1  O_L} \sim L^{-\frac{3}{2} \cdot 2} = L^{-3}$, whereas causation decays as $\braket{O_1 \, G \, O_L} \sim L^{-\frac{15}{2} \cdot 2 + 1} = L^{-14}$! Although this exponent was previously obtained for gSPT edge modes via a similar counting argument \cite{Verresen2021}, here we identify it as a universal property of causation. We confirm both scalings for the choice $O_1 = Y_1 X_2$ in Fig.~\ref{fig:figure1}a, where $\braket{O_1 \; G \; O_L}$ is computed by evaluating the resolvent in the exact free-fermion basis after a Jordan-Wigner transformation (see End Matter).

The same reasoning applies to $\mathbb Z_2$-even, $T$-odd operators, which flow to odd-level descendants of the $h = 0$ tower. The Virasoro character is \cite{DiFrancesco1997} \begin{align}\label{eq:ising_char}
    \chi_0 &\propto 1 + q^2 + q^3 + 2q^4 + 2q^5 + 3q^6 + 3q^7 + 5q^8 \nonumber\\
    &\quad + 5q^9 + 7q^{10} + 8q^{11} + 11q^{12} + 12q^{13} + \cdots
\end{align}Since $\chi_0$ has no $q^1$ term, the first odd-level descendant occurs at level 3 ($h = 3$), while the first odd-level quasiprimary occurs at level 11 ($h = 11$). Correlation therefore decays as $L^{-6}$ and causation as $L^{-21}$, a suppression of fifteen orders of magnitude. We numerically confirm both scalings for the choice $O_1 = Y_1 Z_2 X_3$ in Fig.~\ref{fig:figure1}b, again by computing the resolvent in the free-fermion basis.

\textit{Example: Tricritical Ising CFT---}The tricritical Ising CFT is realized in the spin-1 Blume-Capel chain, \begin{align}
    H = \sum_{n = 1}^L \left[\alpha (S_n^z)^2 + \beta S_n^x\right] - \sum_{n = 1}^{L-1} S_n^z S_{n+1}^z,
\end{align}tuned to its tricritical point, $\alpha = 0.910207$ and $\beta = 0.415685$ \cite{vonGehlen1994}. Like the Ising chain, this model has a $\mathbb Z_2$ spin-flip symmetry, $\prod_n e^{i \pi S_n^x}$, and the antiunitary symmetry $T = K$. At the boundary, $\mathbb Z_2$-even operators flow to the identity tower, whose Virasoro character is $\chi_0 \propto 1 + q^{2} + q^{3} + 2q^{4} + 2q^{5} + 4q^{6} + 4q^{7} + 7q^{8} + 8q^{9} + \ldots$ \cite{DiFrancesco1997}. The first odd-level descendant occurs at level 3 ($h = 3$), and the first odd-level quasiprimary at level 9 ($h = 9$): correlation between $\mathbb Z_2$-even, $T$-odd boundary operators therefore decays as $L^{-6}$, while causation decays as $L^{-17}$. We confirm both scalings for the choice $O_1 = S^y_1 S^z_2 S^x_3$ in Fig.~\ref{fig:figure1}c, where the causation function is obtained by perturbing one end of the chain with $\lambda \, O_1$ ($\lambda = 0.1$) and measuring the response of its mirror image, $\braket{S^x_{L-2} S^z_{L-1} S^y_L}$, at the other end.

\textit{Example: Free-Fermion CFT---}We now consider non-interacting Majorana chains in the BDI class \cite{AltlandZirnbauer1997, Verresen2018},
\begin{align}\label{eq:bdi}
    H = \sum_a t_a H_a, \qquad H_a := \frac{i}{2} \sum_n \tilde{\gamma}_n \gamma_{n+a},
\end{align}where $\gamma_n$ and $\tilde\gamma_n$ are the Majorana operators at site $n$, on which $T = K$ acts as $\gamma_n \to \gamma_n$ and $\tilde\gamma_n \to -\tilde\gamma_n$. This family generalizes the trivial ($H_0$) and Kitaev ($H_1$) chains \cite{Kitaev2001}. At criticality, correlations again decay algebraically; certain causation, however, decays \textit{exponentially}.

We illustrate this at the critical point $t_0 = t_1 = 1$. Here, the lattice flows to the free Majorana CFT, $H  = i\int \chi(x) \partial_x \tilde{\chi}(x) \; dx$, where $\chi$ and $\tilde \chi$ are the continuum limits of $\gamma$ and $\tilde \gamma$ and inherit their $T$-parities. In particular, $\chi$ is $T$-even. Since the left edge of the lattice model terminates on a real Majorana, the corresponding boundary condition is $\tilde \chi(t, 0) = 0$ \cite{Verresen2020}, so all $\partial_t^n \tilde \chi(t, 0)$ vanish. Because lattice fermion correlators Wick-contract, every field in the continuum expansion of a lattice Majorana must be a single-fermion operator. On the $(0+1)$D boundary, these are exhausted by the boundary fermion $\chi(0)$ and its time derivatives, $\partial_t^k \chi(0)$. Note that this implies that a $T$-odd lattice operator flows entirely to odd time derivatives of $\chi$, and hence causation vanishes by Eq.~(\ref{eq:time_descendant}). Adding an irrelevant perturbation such as $\xi \chi(x) \partial_x^2 \tilde \chi(x) $ introduces a length scale but algebraic order still cancels and causation decays exponentially, $\sim e^{-L/\xi}$.

In fact, one can show this holds already at the lattice level. Perturb $(t_0 , t_1) = (1,1)$ with any couplings $\{ t_\alpha \}$ that introduce no imaginary edge modes at the left end. Then any $T$-odd operator is a time derivative of an exponentially localized---in some cases strictly local---lattice operator (End Matter), so causation is exponentially suppressed by Eq.~(\ref{eq:time_descendant}). We confirm this in Fig.~\ref{fig:figure1}d for the critical chain $(t_{-1}, t_0, t_1) = (1, 3, 2)$ by computing the resolvent in the free-fermion basis, where the correlation $\braket{\tilde\gamma_1 \, \gamma_L}$ decays algebraically while the causation $\braket{\tilde\gamma_1 \, G \, \gamma_L}$ decays exponentially. This and the preceding field-theory claims are derived in detail in the End Matter.

\textit{Application: Gapless SPTs---}Causation quantifies finite-size edge-mode splittings in gapless SPTs. To see this, consider a fine-tuned gSPT whose edge modes are decoupled from the critical bulk, and introduce bulk-edge couplings as a perturbation,
\begin{align}\label{eq:edge_coupling}
    H = H_\text{bulk} + \lambda \left( \eta^\text{left}\, O^\text{left} + O^\text{right}\, \eta^\text{right} \right),
\end{align}where $\eta^\text{left,right}$ act on the edge modes and $O^\text{left,right}$ are the coupled bulk terms. Treating the bulk-edge couplings as a perturbation, we obtain the edge-mode splitting at second order \footnote{We assume symmetry forbids $O(\lambda)$ contributions. This can fail, for instance, when the two ends of the chain carry distinct boundary conditions \cite{Prembabu2024}. },
\begin{align}\label{eq:edge_splitting}
    \Delta E \propto \lambda^2 \braket{O^\text{left} \, G \, O^\text{right}},
\end{align}which is precisely a causation function of the bulk. The localization of edge modes in gapless SPTs is therefore determined by the system's capacity to mediate \textit{causal} influence between its boundaries.

We first illustrate this for free fermions, where it was found that the chains~(\ref{eq:bdi}) support exponentially localized edge modes even when gapless~\cite{Verresen2018}, despite their algebraic correlations. We propose a simple explanation of this behavior: the edge-mode splitting is quantified by causation rather than correlation and, by the free-fermion argument above, causation decays exponentially. 

As a concrete example, take the Majorana chain $H = t_1 H_1 + t_2 H_2$ (\ref{eq:bdi}), which hosts exact edge modes $\gamma_1$ and $\tilde \gamma_L$ decoupled from the bulk \cite{Verresen2018}. Now perturb with $t_0 H_0$, which couples the edges. The new terms realize Eq.~\eqref{eq:edge_coupling} with $\eta^\text{left} = \gamma_1$, $\eta^\text{right} = \tilde\gamma_L$, and $O^\text{left} = \tilde\gamma_1$, $O^\text{right} = \gamma_L$, so the edge mode splitting \eqref{eq:edge_splitting} is the bulk causation function $\braket{\tilde\gamma_1 \, G \, \gamma_L}$. This causation decays exponentially in the chain length, $\braket{\tilde\gamma_1 \, G \, \gamma_L} \sim e^{-L/\xi}$, where the length scale $\xi$ is determined by the zeros of $f(z) = t_0 + t_1 z + t_2 z^2$ inside the unit disk. For instance, the family $(t_0, t_1, t_2) = (t_0, 1 + t_0, 1)$ realizes the Majorana CFT \cite{Verresen2018}, where a single inner zero, $z = - t_0$, sets the decay and the edge-mode localization is $\xi = - 1/\ln|t_0|$.

The same argument extends to interacting spin chains, where the edge-mode splitting is more involved but is still quantified by causation. As a novel example, consider the gapless SPT
\begin{align}\label{eq:H}
H &= - \sum_n \left(Z_{n-1}X_{n} Z_{n+1} + Z_{n-1}X_nX_{n+1}Z_{n+2}
    \right),
\end{align}which is an Ising critical point between two phases protected by the $\mathbb Z_2 \times \mathbb Z_2^T$ symmetry generated by $P = \prod_n X_n$ and $T=K$ \cite{Verresen2017}: the first term stabilizes the cluster SPT phase \cite{Suzuki1971, Raussendorf2001, KeatingMezzadri2004, KoppChakravarty2005, Son2011}, whereas the latter term would stabilize a spontaneous breaking down to $PT$, which protects a remaining SPT phase; in both nearby phases, the edge mode is a Kramers pair under $PT$. With free boundaries, a unitary $U$ maps $H$ to the bulk Ising CFT with a decoupled qubit at each end, i.e., $UHU^\dag = -\sum_{n=2}^{L-2} Z_n Z_{n+1} - \sum_{n=2}^{L-1} X_n$; the symmetries become $P \to Z_1 Z_L \tilde P$ and $T \to X_1 X_L \tilde P T$, where $\tilde P = \prod_{n=2}^{L-1} X_n$ is the bulk Ising symmetry. The two edge qubits leave the ground state fourfold degenerate.

To lift this degeneracy, we add generic symmetry-allowed bulk-edge couplings. This realizes Eq.~(\ref{eq:edge_coupling}) with two coupling operators per edge, $\eta_1 O_1$ and $\eta_2 O_2$, where $O_1$ and $O_2$ are both $\mathbb Z_2$-odd and $T$-odd. Thus, their causal contributions flow to odd-level quasiprimaries of the $h=\frac{1}{2}$ tower, the first two of which occur at levels 7 and 9 [Eq.~(\ref{eq:ising_odd})]. The energy splitting in the fourfold degenerate space is given by Eq.~(\ref{eq:edge_splitting}), and splits at two scales: one by $L^{-14}$ from causation of the level-7 quasiprimary, and another smaller one by $L^{-18}$ from the level-9 quasiprimary. Both are confirmed with numerics in Fig.~\ref{fig:figure2}a, where we compute the splittings in the free-fermion basis by evaluating the resolvent.

A second $\mathbb Z_2 \times \mathbb Z_2^T$-protected gapless SPT is given by \begin{multline} \label{eq:H'}
    H' = - \sum_n \big( Z_{n-2}X_{n-1}X_nX_{n+1}Z_{n+2} \\
    + Z_{n-1}X_nX_{n+1}Z_{n+2} \big).
\end{multline}
This model is also an Ising critical point, but the nearby cluster SPT phase has been replaced by a phase where the edge mode is a Kramers pair under \textit{both} $T$ and $PT$ \cite{Verresen2017}. It can be analyzed similarly (End Matter). Here, the protecting symmetry allows $O_1$ and $O_2$ to be $\mathbb Z_2$-even and $T$-odd, flowing to the $h = 0$ tower with first quasiprimaries at levels 11 and 13 [Eq.~(\ref{eq:ising_char})]. Causation predicts energy splittings of $L^{-21}$ and $L^{-25}$! We numerically confirm both splittings in the free-fermion basis, Fig.~\ref{fig:figure2}b.

\begin{figure}[t]
  \centering
  \includegraphics[width=0.9\columnwidth]{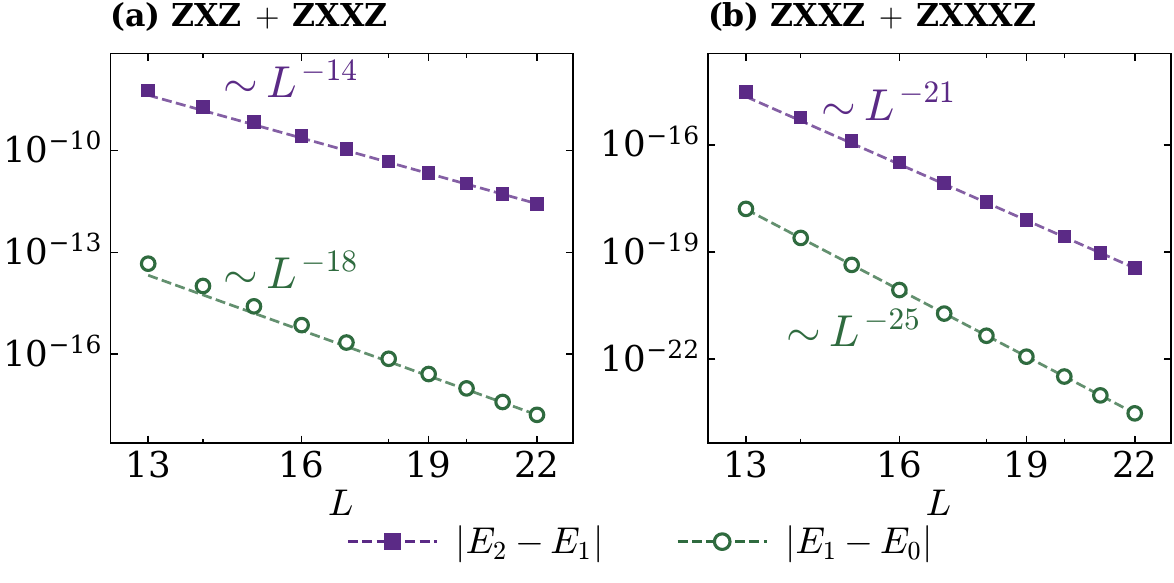}
  \vspace{-1.0em}
\caption{\textbf{Finite-size edge-mode splittings in gapless SPTs.} Starting with fine-tuned gSPTs $H$ and $H'$ with decoupled qubit edge modes, we couple to $O_{1,2}^{\text{left}}$ and $O_{1,2}^{\text{right}}$ of the main text. (a) $H$, with $O_1^\text{left} = Y_2 Y_3 X_4 Y_5$ and $O_2^\text{left} = Y_2 X_3 Y_4 Y_5$ [Eq.~(\ref{eq:H})]; splittings $\sim L^{-14}, L^{-18}$. (b) $H'$, with $O_1^\text{left} = -X_2 Y_3 X_4 Z_5$ and $O_2^\text{left} = Y_2 Y_3 Y_4 Z_5$ [Eq.~(\ref{eq:H'})]; splittings $\sim L^{-21}, L^{-25}$.}
  \label{fig:figure2}
\end{figure}

\textit{Application: Primary sieve in $d \ge 3$ CFTs---}In three- or higher-dimensional CFTs, the conformal algebra is finite-dimensional and descendants are generated by the momentum $P_\mu = i \partial_\mu$, so descendants of a primary $\phi$ have the form $(\partial_{\mu_1})^{k_1} \cdots (\partial_{\mu_i})^{k_i} \phi$ \cite{DiFrancesco1997}. Every descendant of a $(0+1)$D defect or boundary operator is therefore a time derivative and drops out of causation by Eq.~(\ref{eq:time_descendant}). Thus, while the leading correlation of a lattice operator is set by the lowest-dimension field in its continuum expansion, often a descendant, causation is always set by the lowest-dimension primary. In this sense, causation acts as a sieve for primary operators.

In the bulk, or on defects and boundaries of one or more spatial dimensions, spatial derivatives survive but do not spoil the sieve: a lattice operator of indefinite spatial parity \footnote{i.e., has components of both spatial parities; for example, $Z_1 X_2 \propto (Z_1 X_2 + X_1 Z_2) + (Z_1 X_2 - X_1 Z_2)$ decomposes into even- and odd-parity components.} flows to fields of every spatial parity, so no parity rule forbids its primary $\phi$ whenever a spatial descendant $(\partial_x)^n \cdots (\partial_y)^m \phi$ contributes. Since $\phi$ has strictly lower dimension, it dominates.

We demonstrate this in the $(2+1)$D critical transverse-field Ising model, computing corner-to-corner observables on an $L \times L$ lattice with DMRG \cite{White1992, White1993, Hauschild2018, tenpy2024} (Fig.~\ref{fig:combined}a). Convergence was reached for all system sizes shown with bond dimension $\chi = 300$. At a corner, the operators $Z \sim \phi_c$ and $Y \sim \partial_t \phi_c$ are the lowest-dimension $\mathbb Z_2$-odd corner primary and its first descendant. Their correlations give $\Delta_Z \approx 1.8$ and $\Delta_Y \approx 2.8$, in reasonable agreement with the corresponding classical hinge magnetization exponent $\beta_2/\nu \approx 2.03$ computed by Monte Carlo \cite{Pleimling1998}, with the difference presumably attributable to finite-size effects. Correlation resolves only these light operators, since any correlator of $\mathbb Z_2$-odd, $T$-odd operators is dominated by the descendant $\partial_t \phi_c$.

\begin{figure}[t]
  \centering
  \raisebox{-1.8pt}{\includegraphics[trim={1cm 1cm 1cm 1cm}, clip=false,
    width=\columnwidth]{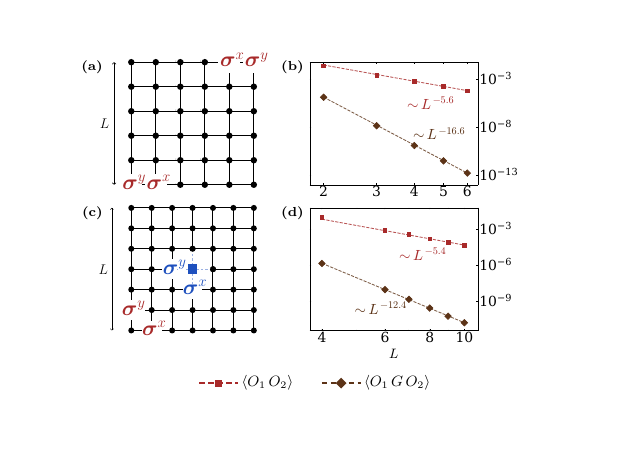}}
\caption{\textbf{Correlation and causation in the $(2+1)$D Ising CFT.}
(a) $L\times L$ Ising lattice with operators $O_1 := \sigma^y_{1,1}\sigma^x_{1,2}$ (bottom left) and $O_2 := \sigma^x_{L,L-1}\sigma^y_{L,L}$ (top right). (b) Decay of $\langle O_1 O_2\rangle$ and $\langle O_1 G O_2\rangle$ with $L$. (c) Lattice with a magnetic line defect (blue square) and $O_2$ at its center, shown for odd $L$; even $L$ differs slightly as it has no center site. Operator definitions are unchanged up to repositioning. (d) Correlation and causation with the defect, where odd-$L$ data is rescaled by a constant factor.}
  \label{fig:combined}
\end{figure}

Causation instead filters out this descendant. Every time derivative drops out, leaving the lowest-dimension $T$-odd corner primary as the leading term. To extract this primary, we perturb the bottom-left corner by $Y_{(1,1)} X_{(1,2)}$, a generic operator carrying the same charges as $Y$, and measure the response $\braket{X_{(L,L-1)} Y_{(L,L)}}$ at the top-right corner with DMRG (Fig.~\ref{fig:combined}a). The causation decays as $L^{-16.6}$ (Fig.~\ref{fig:combined}b), identifying a new $T$-odd, $\mathbb Z_2$-odd corner primary $\phi_c^T$ with $\Delta_{\phi_c^T} \approx 8.8$. This is consistent with mean-field theory, where the lowest-dimension $T$-odd primary is a composite of three $\phi_c$'s and three time derivatives of dimension $3 \cdot [\phi_c] + 3 \cdot [\partial_t] \approx 8.4$. Note that this primary is invisible to correlation: its $L^{-17.6}$ contribution is masked by the $L^{-5.6}$ decay of $\partial_t \phi_c$.

The same filtering mechanism also applies in the presence of conformal defects. Consider a modification of the $L \times L$ lattice, with the center spin pinned downwards to realize a $(0+1)$D magnetic line defect (Fig.~\ref{fig:combined}c), which breaks $\mathbb Z_2$ but preserves $T$ \cite{Allais2014, ParisenToldin2017, Cuomo2022, Hu2024}. The connected correlation of $Z$ between the corner and the defect decays as $L^{-3.4}$, consistent with the corner primary $\Delta_{\phi_c} \approx 1.8$ plus the leading defect primary $\phi_d$ of dimension $\Delta_{\phi_d} \approx 1.6$ \cite{Hu2024}. For the $T$-odd operator $YX$, correlation resolves only the first descendants of these fields, $\partial_t \phi_c$ and $\partial_t \phi_d$, decaying as $L^{-5.4}$ with $5.4 = (1.8 + 1) + (1.6 + 1)$. Causation, computed as before with DMRG (up to $\chi = 1800$), filters out these descendants, decaying as $L^{-12.4}$ (Fig.~\ref{fig:combined}d). Subtracting the corner primary dimension identifies a $T$-odd defect primary $\phi_d^T$ with $\Delta_{\phi_d^T} \approx 13.4 - 8.8 = 4.6$, in excellent agreement with the heaviest magnetic line defect primary found with the fuzzy sphere method, $\Delta = 4.64(14)$ \cite{Hu2024}.

\textit{Discussion---}We have shown that correlation and causation can exhibit strikingly different scaling at criticality. Since time-derivative fields make no contribution, the causation of a generic lattice operator is governed by the lowest-dimension quasiprimary in its expansion, and in $d\ge3$, the lowest primary. This implies causation is a good probe of primary operator dimensions and explains the sharp localization of edge modes in gapless SPTs. 

The former indicates that causation is a practical tool for higher-dimensional CFTs, where much of the operator spectrum is unknown. By computing causation in different symmetry sectors, one can isolate the corresponding lowest primary, as demonstrated here for $\phi_c^T$ and $\phi_d^T$. The same approach extends to other defects, boundaries, and critical models. In each case, causation requires only a ground state calculation in the presence of a perturbation, so it is easily computed with exact diagonalization, DMRG, quantum Monte Carlo, or any numerical method that is already suited to compute correlation. A systematic study of causation across symmetry sectors of $d \ge 3$ lattice CFTs may uncover additional heavy bulk, defect, and boundary primary operators, and we leave this to future work.

Our results also raise several new questions. We have defined causation as the static ($\omega \to 0$) limit of the Kubo response formula; at finite frequency, descendant fields may again contribute, and to what extent the suppression persists would be worthwhile to investigate. In addition, the suppression mechanism identified here relies on $\mathbb Z_2^T$ time-reversal symmetry to isolate time-derivative fields. It will be interesting to determine whether other symmetries play a similar role in suppressing causation.

\textit{Acknowledgments---}The authors thank Aashish Clerk and Curt von Keyserlingk for insightful discussions, and in particular Saranesh Prembabu for suggesting a defect-based example. C.W. was supported in part by a Selove Summer Research Scholarship from the University of Chicago's Department of Physics. DMRG calculations were performed using the TeNPy library (version 1.1.0) \cite{tenpy2024}.

\bibliography{draft}

\onecolumngrid
\begin{center}
\textbf{End Matter}
\end{center}
\vspace{2pt}
\twocolumngrid

\setcounter{equation}{0}
\renewcommand{\theequation}{A\arabic{equation}}
\makeatletter
\renewcommand{\theHequation}{A\arabic{equation}}
\makeatother

\textit{Free Fermions---}As mentioned in the main text, lattice fermions flow to single fermionic operators in the low-energy theory. To see this, suppose we have a lattice model of free fermions which flows to a scale-invariant field theory. The lattice Majorana fermions $c(x)$ can be expressed as a sum of continuum scaling operators $c(x) \sim \sum_i a_i O_i(x)$ where each $O_i$ is a (possibly composite) fermionic operator of dimension $\lambda_i$. Here, $\sim$ indicates that all correlation functions $\langle c(x_1) \cdots c(x_n)\rangle$ are asymptotic to the corresponding continuum correlation functions as all $|x_i-x_j|$ become large. Since the lattice fermions satisfy Wick contraction, so too must the continuum correlation functions order-by-order in the scaling dimension. In particular, if we take all $\lambda_i$ to be distinct, each $O_i(x)$ satisfies Wick expansion with itself and all other $O_j(y)$. However, since the Majorana CFT is itself free, correlations of composite $O_i$'s may also be computed by Wick-contracting their constituent fermions, which yields partial contractions between composites which are absent in the former expansion (unless each $O_i$ is a single fermion). Consistency therefore requires any boundary $O_i$ to be of the form $\partial_t^n \chi$ or $\partial_t^n \tilde \chi$, which exhausts all single-fermion boundary operators. Imposing $\tilde \chi(t,0) = 0$ as in the main text, the only non-vanishing $T$-odd operators are of the form $\partial_t^k \chi(t,0)$ for odd $k$ and hence exhibit vanishing causation.

Let us introduce an irrelevant perturbation, $\xi \, \chi(x) \, \partial_x^2 \tilde \chi(x)$, so that the Hamiltonian on a chain of length $L$ becomes $H = i \int_0^L  \chi(x) \left( \partial_x - \xi \partial_x^2 \right) \tilde{\chi}(x) \, dx$. In the unperturbed theory, the leading $T$-odd boundary operator, $\partial_x \tilde \chi(0) \propto \partial_t \chi(0)$, is the time derivative of a strictly local fermion and therefore has vanishing causation. The perturbation changes this conclusion because $\partial_x\tilde{\chi}(0)$ is no longer the time derivative of a strictly local operator, but rather an exponentially localized one. To see this, write \begin{align}
    \partial_x \tilde \chi(0) = i [H, \chi'] + e^{-L/\xi} \partial_x \tilde \chi(L),
\end{align}where $\chi'=\frac{1}{2\xi}\int_0^L e^{-x/\xi}\chi(x) \,dx.$ For the first term, Eq.~\eqref{eq:time_descendant} reduces the causation to the correlation overlap $i \braket{[H - E_0, \chi'] \; G \; \partial_x \chi(L)} = i\braket{\chi'\,\partial_x\chi(L)}$, which picks up only the tail $e^{-L/\xi}$ of $\chi'$ at the far edge. The second term is suppressed at the same rate by its prefactor, $e^{-L/\xi}$. It follows that the end-to-end causation $\braket{\partial_x \tilde \chi(0) \; G \; \partial_x \chi(L)}$ decays as $e^{-L/\xi}$ with decay length $\xi$ set by the strength of the irrelevant perturbation.

As we discuss in the main text, this result holds more generally on the lattice.  Take a chain of length $L$ in the BDI class (\ref{eq:bdi}). Its associated Laurent polynomial is defined as $f(z) := \sum_a t_a z^a$. Let $N_z$ denote the number of zeros of $f$ inside the unit circle, $N_p$ its pole order at the origin, and $\omega := N_z - N_p$ the `winding number' \cite{Verresen2018}.

\begin{theorem}
    Consider any BDI Majorana chain with $\omega \ge 0$. As $L \to \infty$,
    \begin{enumerate}
        \item If $N_p > 0$, then $\braket{\tilde\gamma_1 \, G \, \gamma_L} \lesssim e^{-L/\xi}$.
        \item Otherwise, $\braket{\tilde \gamma_1 \; G \; \gamma_L} = 0$.
    \end{enumerate}Here, $\xi = -1/\ln|z_*|$, with $z_*$ the largest zero of $f$ inside the unit circle.
\end{theorem}
Here we use the convention that $G$ projects out the ground state subspace.
The proof follows from writing $\tilde\gamma_1$ as the time derivative of an exponentially localized operator in case one, and of a strictly local operator in case two. Exponential decay and vanishing then follow from Eq.~\eqref{eq:time_descendant}.

\emph{Proof.}
We seek an operator $\mathcal{O} = \sum_{a\ge1}\alpha_a\gamma_a$ with $i[H,\mathcal{O}]=\tilde\gamma_1$ on the half-infinite chain $H = \frac{i}{2} \sum_{a} t_a \sum_{n,\,n+a \ge 1}  \tilde  \gamma_n \gamma_{n+a}$. From Eq.~\eqref{eq:bdi},
$i[H,\gamma_a] = -\sum_b t_b\,\tilde\gamma_{a-b}$, so the requirement is
\begin{equation}
    \sum_{a\ge1} t_{a-n}\,\alpha_a = -\delta_{n,1},
    \label{eq:em_cond}
\end{equation}for all $n\ge 1$. If $N_p=0$, $f$ is a polynomial with lowest power $a_0\ge0$, and Eq.~\eqref{eq:em_cond} is solved by the strictly local operator $\mathcal{O} = -\gamma_{a_0+1}/t_{a_0}$. Since $\mathcal O$ is strictly local, $\tilde\gamma_1 = i[H,\mathcal{O}]$ also holds on any finite chain with $L>a_0+1$.

If $N_p>0$, take instead
\begin{equation}
    \alpha_a = \sum_{s=1}^{N_z} \lambda_s\, z_s^{\,a},
    \label{eq:em_ansatz}
\end{equation}
with $z_s$ the zeros of $f$ inside the unit circle (for simplicity of notation we consider the generic case of non-degenerate roots). Since
$\sum_{a\ge1} t_{a-n} z^a = z^n f(z) - \mathds{1}_{n \le N_p} \sum_{k=n}^{N_p} t_{-k} z^{\,n-k}$ and
$f(z_s)=0$, Eq.~\eqref{eq:em_cond} trivially holds for $n>N_p$ and, for $n \le N_p$, reduces to the $N_p$ conditions $\sum_s M_{ns}\lambda_s = \delta_{n,1}$ with $M_{ns} = \sum_{k=n}^{N_p} t_{-k} z_s^{\,n-k}$. This matrix row reduces to a Vandermonde matrix in $z_s^{-1}$, so a solution exists precisely when $N_z \ge N_p$, i.e.\ $\omega\ge0$. The coefficients are bounded by $|\alpha_a| \le \Lambda e^{-a/\xi}$ with $\Lambda=\sum_s|\lambda_s|$. Truncating to a  finite chain of length $L$, $\mathcal{O} = \sum_{a=1}^{L}\alpha_a\gamma_a$, and the commutator becomes $i[H,\mathcal{O}] = \tilde\gamma_1 + E$, which includes the remainder term
\begin{equation}
    E = \!\!\sum_{j>L-q}\!\Big(\sum_{a>L} t_{a-j}\,\alpha_a\Big)\tilde\gamma_j, \label{eq:em_error}
\end{equation}
where $q$ is the highest power of $f(z)$. Note that $\|E\| \le c\, e^{-L/\xi}$ for a constant $c$ set by the Hamiltonian couplings.

For $N_p>0$, Eq.~\eqref{eq:time_descendant} applies with $C = \mathcal{O}$.
The projector term vanishes by fermion parity, and anti-hermiticity forces $\braket{\gamma_a\gamma_L} = \delta_{aL}$, so only the overlap
of $\mathcal{O}$ with the far edge survives,
\begin{equation}
    \braket{\tilde\gamma_1\, G\, \gamma_L}
    = i\braket{\mathcal{O}\gamma_L} - \braket{E\, G\, \gamma_L}
    = i\,\alpha_L - \braket{E\, G\, \gamma_L}.
    \label{eq:em_reduction}
\end{equation}
We can bound $|\braket{E\,G\,\gamma_L}|
\le \|E\|\,\|G\|$ with $\|G\|$ the inverse energy gap, polynomial in $L$ at criticality. Since $\|E \|$ and $|\alpha_L |$ are both bounded by $e^{-L/\xi}$, causation decays exponentially up to polynomial prefactors, giving case one. For $N_p=0$, Eq.~\eqref{eq:em_reduction} holds with $E=0$ and $\alpha_L=0$, which gives case two. \qed

This is consistent with our field-theory result. At $t_0 = t_1 = 1$ the chain realizes the Majorana CFT, and since $N_p = 0$, causation vanishes as predicted in the continuum. After turning on a perturbation such as $t_{-1}$, causation decays exponentially, even if we tune the other coefficients to stay at criticality.

\setcounter{equation}{0}
\renewcommand{\theequation}{B\arabic{equation}}
\makeatletter
\renewcommand{\theHequation}{B\arabic{equation}}
\makeatother

\textit{Interacting gapless SPTs---}Here we derive the ground-state splittings of the two gapless SPTs of the main text. With open boundaries, the first model, $H$, consists of all terms supported on a finite chain of length $L$,
\begin{align}
H = -\sum_{n=2}^{L-1} Z_{n-1}X_{n} Z_{n+1}
    - \sum_{n=2}^{L-2} Z_{n-1}X_nX_{n+1}Z_{n+2}.
\end{align}
A unitary $U = R_x\, U_{\rm CZ}$ relates this to the Ising chain, where $U_{\rm CZ} = \prod_{n=1}^{L-1}\mathrm{CZ}_{n,n+1}$ applies a controlled-Z gate to each bond and $R_x = e^{-i\frac{\pi}{4}\sum_n X_n}$ is a global $\pi/2$ rotation about the $x$-axis. This maps $H$ to a bulk critical Ising chain with sites $1$ and $L$ decoupled. The two free edge qubits leave the ground state fourfold degenerate. To lift the degeneracy, we perturb by the symmetry-allowed edge terms $\lambda(Z_1 Z_2 X_4 + X_1 Z_2 X_3 Z_5)$ and their reflections at the right edge, which in the transformed frame realize Eq.~(\ref{eq:edge_coupling}) with two coupling operators per edge,
\begin{align}\label{eq:pert_a}
    UHU^\dag + \lambda \big( Y_1 O_1^\text{left}
    + X_1 O_2^\text{left} + \cdots \big).
\end{align}Up to signs, $O_1^\text{left} = Y_2Y_3X_4Y_5$ and $O_2^\text{left} = Y_2X_3Y_4Y_5$, both $\mathbb Z_2$-odd and $T$-odd. At second order in $\lambda$, the ground-state manifold splits with eigenvalues
\begin{align}\label{eq:eigenvalues}
    E_{1,2,3,4} = \pm 2 \lambda^2 \sqrt{\lVert M \rVert_F^2 \pm 2 \det M},
\end{align} where $M_{ij} = \braket{O_i^\text{left}\, G\, O_j^\text{right}}$ is the $2\times2$ matrix of bulk causation functions and $\lVert M \rVert_F^2 = \sum_{ij} M_{ij}^2$. For
$|\det M| \ll \lVert M \rVert_F^2$, the manifold splits at two scales, $\lVert M \rVert_F$ and $\det M / \lVert M \rVert_F$.

Since $O_1$ and $O_2$ are $\mathbb Z_2$-odd and $T$-odd, their causal contributions flow to odd-level quasiprimaries of the $h = \tfrac12$ tower, the first two at levels 7 and 9 [Eq.~(\ref{eq:ising_odd})]. Every entry of $M$ is dominated by the level-7 quasiprimary, so the larger splitting is $\lVert M \rVert_F \sim L^{-14}$. Since level 7 contains a single quasiprimary, both operators reduce to the same field at this level, so the two products in $\det M = M_{11}M_{22} - M_{12}M_{21}$ are equal and cancel. The determinant is first nonzero at the order of the level-9 quasiprimary,
giving the smaller splitting $\det M / \lVert M \rVert_F \sim L^{-18}$. Both scalings are confirmed numerically in Fig.~\ref{fig:figure2}a.

For the second model with open boundaries, we add the edge terms $H_\text{edge} = Y_1Z_2Y_3Z_4 + Z_{L-3}Y_{L-2}Z_{L-1}Y_L$, which are symmetry-allowed and commute with all other terms, leaving one protected qubit per edge,
\begin{align}
    H' = &-\sum_{n=3}^{L-2} Z_{n-2}X_{n-1}X_nX_{n+1}Z_{n+2} \nonumber\\
    &-\sum_{n=2}^{L-2} Z_{n-1}X_nX_{n+1}Z_{n+2} - H_\text{edge}.
\end{align}As before, $ZXXZ$ alone stabilizes the phase spontaneously breaking the symmetry down to $PT$, which protects a remaining gapped SPT. The $ZXXXZ$ term stabilizes an SPT nontrivial under both $T$ and $PT$ individually. A unitary $U' = W\, U_{\rm CZ}\, V\, U_{\rm CZ}$ again maps $H'$ to the bulk Ising chain with sites $1$ and $L$ decoupled, where $V = \prod_n (Y_n + Z_n)/\sqrt{2}$ exchanges $Y$ and $Z$, and $W = e^{i\frac{\pi}{4}(Z_2 + Z_{L-1})}$ rotates sites $2$ and $L-1$ about the $z$-axis. Under $U'$, the protecting symmetries transform as $P \to \tilde P$ and $T \to Y_1 Y_L \tilde P T$. Perturbing by $\lambda(X_1 Z_3 Z_4 + Z_2 X_4 Z_5)$ and their reflections at the right edge gives
\begin{align}\label{eq:pert_b}
    U'H'U'^\dag + \lambda \big( Y_1 O_1^\text{left}
    + Z_1 O_2^\text{left} + \cdots \big),
\end{align}where, up to signs, $O_1^\text{left} = X_2Y_3X_4Z_5$ and $O_2^\text{left} = Y_2Y_3Y_4Z_5$, both $\mathbb Z_2$-even and $T$-odd. Their causal contributions flow to odd-level quasiprimaries of the $h = 0$ tower, the first two at levels 11 and 13 [Eq.~(\ref{eq:ising_char})], and Eq.~(\ref{eq:eigenvalues}) applies unchanged: $\lVert M \rVert_F \sim L^{-21}$ from level 11 and $\det M / \lVert M \rVert_F \sim L^{-25}$ from level 13, confirmed in Fig.~\ref{fig:figure2}b.

\setcounter{equation}{0}
\renewcommand{\theequation}{C\arabic{equation}}
\makeatletter
\renewcommand{\theHequation}{C\arabic{equation}}
\makeatother

\textit{Time reversal---}In the main text, we note that $T$-charge flips between successive descendant levels of a boundary tower in a $(1 + 1)$D system. Here we derive this, starting in the bulk. We work in Euclidean coordinates $z, \bar z = x \pm iy$, with $x$ time and $y$ space. Time reversal sends $x \to -x$, so it acts on the coordinates as $T z T = -\bar z$ and $T \partial_z T = -\partial_{\bar z}$. Its action on the geometric Virasoro generators $L_n = z^{1+n}\partial_z$ and $\bar L_n = \bar z^{1+n}\partial_{\bar z}$ is therefore
\begin{align}\label{eq:time_reversal}
    T L_n T = (-1)^n \bar L_n, \qquad T \bar L_n T = (-1)^n L_n.
\end{align}
Now, consider a boundary at $y = 0$ parallel to the time direction. Cardy's gluing condition, $T(z) = \bar T(\bar z)$ \cite{Cardy2004}, leaves a single Virasoro algebra with geometric generators given by the symmetric combination $ L_n + \bar L_n = z^{1+n}\partial_z + \bar z^{1+n}\partial_{\bar z}$. By Eq.~(\ref{eq:time_reversal}), $T (L_n + \bar L_n) T = (-1)^n (L_n + \bar L_n)$, so a level-$k$ descendant carries $T$-charge $(-1)^k$ relative to its primary, which implies $T$-charge alternates level-by-level.

\setcounter{equation}{0}
\renewcommand{\theequation}{D\arabic{equation}}
\makeatletter
\renewcommand{\theHequation}{D\arabic{equation}}
\makeatother

\textit{Free-fermion evaluation of causation---}Here we detail the free-fermion evaluation of causation in the critical Ising chain. A Jordan-Wigner transformation,
\begin{align}\label{eq:jw}
    \gamma_n &= (-1)^n (X_1 \cdots X_{n-1})\, Z_n \nonumber \\
    \tilde\gamma_n &= (-1)^n (X_1 \cdots X_{n-1})\, Y_n,
\end{align}
maps the critical Ising chain to the critical Majorana chain, $H = i\sum_n ( \tilde\gamma_n\gamma_n + \tilde\gamma_n\gamma_{n+1})$. To diagonalize the Hamiltonian, we introduce new Majorana modes
\begin{align}\label{eq:modes}
    \eta_k &= \frac{2}{\sqrt{2L+1}}\sum_{n=1}^{L}\sin\!\big((2n\!-\!1)k\big)\,\gamma_n \nonumber \\
    \tilde\eta_k &= \frac{2}{\sqrt{2L+1}}\sum_{n=1}^{L}\sin(2nk)\,\tilde\gamma_n,
\end{align}with $k = \pi a/(2L+1)$, $a = 1,\dots,L$, which satisfy the Majorana anti-commutation relations $\{\eta_k,\eta_q\} = \{\tilde\eta_k,\tilde\eta_q\} = 2\delta_{kq}$ and $\{\eta_k,\tilde\eta_q\} = 0$. Inverting Eq.~(\ref{eq:modes}) and substituting into $H$ yields $H = 2 i \sum_k \cos (k) \tilde \eta_k \eta_k$. Pairing each $\eta_k,\tilde\eta_k$ into a complex fermion, $c_k = (\eta_k - i\tilde\eta_k)/2$ and $c_k^\dagger = (\eta_k + i\tilde\eta_k)/2$, gives \begin{align}
    H = 4\sum_k \cos(k)\big(c_k^\dagger c_k - \tfrac12\big).
\end{align}Since $\cos(k) > 0$ for all $k$, the ground state $\ket{\psi_0}$ is the vacuum of the $c_k$.

Under the Jordan-Wigner map, Eq.~(\ref{eq:jw}), each operator in $\braket{O_1 G O_L}$ can be written with Majoranas. For example, $O_1 = Y_1 X_2$ (Fig.~\ref{fig:figure1}a) maps to $i\,\tilde\gamma_1\tilde\gamma_2\gamma_2$. Inverting Eq.~(\ref{eq:modes}) then expresses $O_1$ and $O_L$ in the mode operators $c_k$. Since $G$ is diagonal in the corresponding Fock basis, $\braket{O_1 G O_L}$ can be evaluated numerically.

\end{document}